\documentclass{article}
\usepackage[utf8]{inputenc}
\usepackage{amsthm,amsmath}
\usepackage{float}

\usepackage{tikz}
\usetikzlibrary{decorations.pathmorphing}

\theoremstyle{plain}
\newtheorem{theorem}{Theorem}
\newtheorem{lemma}[theorem]{Lemma}
\newtheorem{proposition}[theorem]{Proposition}
\newtheorem{property}[theorem]{Property}
\newtheorem{corollary}[theorem]{Corollary}

\theoremstyle{definition}

\theoremstyle{remark}
\newtheorem{observation}[theorem]{Observation}

\newcommand{\LL}{\mathcal{L}}
\newcommand{\HH}{\mathcal{H}}
\newcommand{\dBE}{de Bruijn-Erd\H{o}s property}

\newcommand{\Line}{\overline}

\tikzstyle{vertex}=[circle,fill=black,minimum size=3pt,inner sep=0pt]
\tikzstyle{spath}=[decorate, decoration={snake, segment length=1mm, amplitude=.5mm}]

\title{Graphs with no induced house nor induced hole have the de
  Bruijn-Erd\H{o}s property \footnote{Partially supported by Basal
    program AFB170001, CONICYT Fondecyt/Regular 1180994, and
    programs ANR-17-CE40-0015 and ANR-19-CE48-0016 from the French
    National Research Agency (ANR).}  }

\author{Pierre Aboulker$^1$, Laurent Beaudou$^2$, Martín Matamala$^{3}$ and José Zamora$^{4}$\\
\small ($1$) DIENS, \'Ecole normale sup\'erieure, CNRS, PSL University, Paris, France\\
\small ($2$) Higher School of Economics, Moscow, Russian Federation\\
\small ($3$) DIM-CMM, Universidad de Chile, Santiago, Chile \\
\small ($4$) Depto. Matemáticas, Universidad Andres Bello, Santiago, Chile}

\date{\today}

\begin{document}

\maketitle

\begin{abstract}
  A set of $n$ points in the plane which are not all collinear
  defines at least $n$ distinct lines. Chen and Chvátal
  conjectured in 2008 that a similar result can be achieved in the
  broader context of finite metric spaces. This conjecture remains
  open even for graph metrics. In this article we prove that
  graphs with no induced house nor induced cycle of length at
  least~5 verify the desired property. We focus on lines generated
  by vertices at distance at most 2, define a new notion of ``good
  pairs'' that might have application in larger families, and
  finally use a discharging technique to count lines in
  irreducible graphs.
\end{abstract}

\section{Introduction}
Given a set of $n$ points in the Euclidean plane, they are all
collinear or they define at least $n$ distinct lines. This result
is a corollary of Sylvester-Gallai Theorem (suggested by Sylvester
in the late nineteenth century~\cite{sylvester_1893} and proved by
Gallai some forty years later as reported by
Erd\H{o}s~\cite{erdos_1982}).

Can this property of the Euclidean plane be satisfied by more
general metric spaces? We first need to specify the notion of line
in a general metric space $(V,d)$. We say that a point $z$ in $V$
is {\em between} $u$ and $v$ (points in $V$) if
$d(u,v)=d(u,z)+d(z,v)$. Given two points $u$ and $v$, the set of
points between $u$ and $v$ is the {\em interval} defined by $u$
and $v$, denoted $I(u,v)$. Note that $u$ and $v$ are in $I(u,v)$.
When the order is not relevant, we may say that three points $u,v$
and $z$ are {\em collinear}. This means that one of them is
between the others. The {\em line} defined by two points $u$ and
$v$ is the set of all points $z$ such that $u,v$ and $z$ are
collinear. It is denoted by $\Line{uv}$. A line is {\em universal}
if it equals $V$. In this wording, Sylvester-Gallai Theorem states
that $n$ points in the Euclidean plane define at least $n$
distinct lines or form a universal line. In 1948, de Bruijn and
Erd\H{o}s studied a combinatorial
problem~\cite{debruijn_erdos_1948} implying Sylvester-Gallai
Theorem. This explains partly the name of the following property.

A metric space $M=(V,d)$ satisfies the {\em de Bruijn-Erd\H{o}s
  property} if \begin{equation}
  \label{con:dBE}
  \tag{DBE} M \text{ has a universal line,
    or at least } |V| \text{ distinct lines.}
\end{equation}

In 2008, Chen and Chvátal~\cite[Question 1]{chen_chvatal_2008}
wondered if all finite metric spaces satisfy the \dBE. By lack of
counterexample, this question has now grown to be a conjecture:
the Chen-Chvátal conjecture. A \emph{graph metric} is a metric
space that arises from a graph: the ground set is the set of
vertices and the distance between two vertices corresponds to the
length (number of edges) of a shortest path linking these
vertices. Chen-Chvátal conjecture remains open even for graph
metrics. In recent years, it has been proved that several families
of metric spaces satisfy the \dBE: every metric space with
distances in $\{0,1,2\}$
\cite{chiniforooshan_chvatal_2011,chvatal_2014}; graph metrics
induced by graphs which are chordal~\cite{beaudou_bondy_2015} or
distance hereditary~\cite{aboulker_kapadia_2015}. More generally,
any graph metric defined by a graph $G$ such that every induced
subgraph of $G$ is either a chordal graph, has a cut-vertex or a
non-trivial module~\cite{aboulker_matamala_2018}. Several
strengthenings of the initial conjecture have been
suggested~\cite{matamala_zamora_2020}. For a good overview of
previous results and open problems, one may read the enjoyable
survey written by Chvátal in 2018~\cite{chvatal_2018}.

In this paper, we prove that graph metrics for so-called \{house,
hole\}-free graphs satisfy the \dBE. The {\em house} is the graph
on five vertices obtained by adding one chord to a 5-cycle. A {\em
  hole} is a cycle on at least five vertices. The class of
\{house, hole\}-free graphs consists in those graphs that do not
admit a house or a hole as an induced subgraph. The main result is
thus the following theorem.

\begin{theorem}
  \label{thm:main}
  The class of \{house, hole\}-free graphs satisfies the \dBE.
\end{theorem}

It actually answers Problem 3 of Chvátal's
survey~\cite{chvatal_2018}. From now on, we let $\HH$ denote the
class of \{house, hole\}-free graphs.

\paragraph{The proof in a nutshell}
We shall prove Theorem~\ref{thm:main} by induction.  As usual for
inductive proofs, we need to adjust the induction hypothesis since
it is both what we want to prove (thus we would fancy a weak
statement to lighten the proof) and our hypothesis for proving
(thus we look for a bold and strong statement to ease the
deductive process). In our case, we strengthen the original
statement by only considering lines generated by vertices at
distance at most~2 from each other. We prove that any graph $G$ in
$\HH$ on $n$ vertices satisfies the following property.
\begin{equation}
  \label{eq:prop}
  \tag{DBE-2} G \text{ has a universal line }\Line{uv} \text{,
    with }d_G(u,v)\leq 2\\ \text{, or at least } n \text{ lines.}
\end{equation}
To that end, we study two families of lines: lines generated by
pairs at distance~1 (genuinely called $\LL_1$), and lines
generated by {\em some} pairs at distance exactly~2 (similarly
called $\LL_2$). We focus on those pairs that generate the same
line. For lines in $\LL_1$ such pairs form a complete bipartite
graph (Section~\ref{sec:l1}). For lines in $\LL_2$, they are
mostly arranged in a star manner (in other words they have a
\\
``center'') except if there is a $C_4$-module in the graph
(Section~\ref{sec:l2}). Then we observe that these two families
are disjoint when there is no universal line
(Section~\ref{sec:l1l2}).  Finally, after proving that a minimal
counter-example to \eqref{eq:prop} cannot have a $C_4$-module
(Section~\ref{sec:noC4}), the last part of the proof deals with
the actual counting of lines for graphs in $\HH$ with no
$C_4$-module (Section~\ref{sec:count}). By use of discharging
techniques (which is nothing but a sophisticated double counting
argument) we give a weight of~1 to every line in $\LL_1$ and
$\LL_2$ and distribute these weights to vertices (given to the
center of the star for lines in $\LL_2$ and split into two halves
for lines in $\LL_1$). Finally we show that every vertex has
received at least a weight of~1 after this process. Thus, the
number of lines is no less than the number of vertices.

\section{Preliminaries, notations, previous work}

In this section, we introduce the tools needed for a smooth
understanding of the proof. All considered graphs are simple,
finite and connected. We assume basic knowledge in graph
terminology. Let us specify that a {\em $C_4$-module} in a graph
$G$ is an induced subgraph isomorphic to a 4-cycle and such that
every other vertex of $G$ is either complete or anticomplete to
these four vertices.

\subsection{Pairs generating the same line}
\label{sub:aboulker}

The main issue for us is when many pairs of vertices generate the
same line. To that matter, we shall make heavy use of a recent
result by Aboulker, Chen, Huzhang, Kapadia and
Supko~\cite{aboulker_chen_2016}. They describe the structure of
pairs generating the same line and formalize their result in the
framework of pseudometric betweenness~\cite[Section
  6]{aboulker_chen_2016}. This framework is not our focus so let
us rephrase some of their results for the case of graphs.

Let $G$ be a graph. A sequence $(a,b,c,d)$ of four distinct
vertices of $G$ forms a {\em parallelogram} if $b \in I(a,c), c
\in I(b,d), d \in I(c,a)$ and $a \in I(d,b)$.  Now let $uv$ and
$xy$ be two pairs of vertices of $G$ (not necessarily disjoint
pairs). Authors in~\cite{aboulker_chen_2016} define three types
of relation. Pairs $uv$ and $xy$ are said:

\begin{itemize}
\item in {\em $\alpha$-relation} if there is a shortest path in
  $G$ containing $\{u,v,x,y\}$\footnote{This set may have order
  3.},
\item in {\em $\beta$-relation} if $(u,v,x,y)$ or $(u,v,y,x)$
  forms a parallelogram and $d_G(u,v)=d_G(x,y)=1$,
\item in {\em $\gamma$-relation} if $(u,x,v,y)$ forms a parallelogram
  and $\Line{uv}=I(u,v) = I(x,y) = \Line{xy}$.
\end{itemize}

\begin{theorem}[{Rephrasing of~\cite[Lemma 6.6]{aboulker_chen_2016}}]
    \label{thm:eqlines}
    In a graph $G$, given two pairs of vertices $xy$ and $uv$, if
    lines $\Line{xy}$ and $\Line{uv}$ are equal, then those pairs
    are $\alpha$-related, $\beta$-related, or $\gamma$-related.
\end{theorem}

We will also need this easy-to-prove property of parallelograms.

\begin{property}[{Rephrasing of~\cite[Lemma 6.9]{aboulker_chen_2016}}]
\label{prop:parallelogram}
If $(a,b,c,d)$ is a parallelogram in a graph $G$, then $d_G(a,b)
= d_G(c,d)$, $d_G(a,d) = d_G(b,c)$, and $d_G(a,c) =
d_G(b,d)$.
\end{property}

\subsection{Sets of lines $\LL_1$ and $\LL_2$}

Since Property~\eqref{eq:prop} is focused on lines generated by
pairs at distance at most 2, we naturally specify some subfamilies
of lines. Namely, for a graph $G$ in $\HH$, we define the set of
lines induced by vertices at distance exactly 1,
\begin{equation*}
  \LL_1(G) := \{\Line{uv}: uv \text{ is an edge} \}.
\end{equation*}
It turns out that the graph is always clear from the context in
this paper. Thus we shall abusively write $\LL_1$ instead of
$\LL_1(G)$.

The other family of lines that we shall consider is the set of
lines generated by {\em some} pairs at distance exactly 2. Note
that Property~\eqref{eq:prop} refers to {\em all} lines generated
by pairs at distance at most 2. And we actually need all of them
to prove that a minimal counter-example has no $C_4$-module. But
in the final counting (for those potential counter-examples), we
rely on a yet smaller subset of them: the lines generated by {\em
  good pairs}. A {\em good pair} is a pair of vertices $u$ and $v$
at distance exactly 2 and such that they have some common
neighbour $z$ satisfying $\Line{uz} = \Line{zv}$ (observe that
$\Line{uz}$ is in $\LL_1$). For any graph $G$ in $\HH$ we define
the set
\begin{equation*}
  \LL_2(G) := \{\Line{uv}: uv \text{ is a good pair} \}.
\end{equation*}
And we shall write $\LL_2$ since the context is always clear. Let
us repeat, for there is no better way to insist, that in general
$\LL_2$ is not the set of all lines generated by pairs of vertices
at distance exactly 2. It is a subset of it.

\subsection{Useful lemmas}

In an attempt to make subsequent proofs lighter, we gather here a
few results. First we study the structure of good pairs. Some of
these results are very general and could be used in larger
frameworks than \{house, hole\}-free graphs. Then, we give some
structural results about \{house, hole\}-free graphs. Their proof
is neither hard nor technical and a simple drawing makes things
pretty obvious. Still we wrote it all down. The reader is engaged
to skip the proofs if the result is clear enough. 

\subsubsection{On good pairs}

We start with an easy and general observation deriving from
triangular inequality and the definition of collinearity.

\begin{observation}
  In any connected graph $G$ and for any three vertices $u,v$ and
  $z$, if $z$ is not in $\Line{uv}$ then distances $d_G(z,u)$
  and $d_G(z,v)$ differ by at most $d_G(u,v) - 1$.
\end{observation}

Note that if $uv$ is an edge, it means that $z$ is equidistant to
$u$ and $v$. The following lemma asserts it is also the case when
$uv$ is a good pair.

\begin{lemma}
  \label{lem:distgood}
  In any connected graph $G$, if $uv$ is a good pair and a vertex
  $z$ is not in line $\Line{uv}$, then $d_G(z,u) = d_G(z,v)$.
\end{lemma}

\begin{proof}
  Since $z$ is not in $\Line{uv}$ and $d_G(u,v) = 2$, distances
  from $z$ to $u$ and $v$ differ by at most~1. Suppose for a
  contradiction that they differ by exactly~1. Without loss of
  generality, we may assume that $d_G(z,u) = k$ and $d_G(z,v) =
  k+1$. Since $uv$ is a good pair, there is a vertex $c$ between
  $u$ and $v$ such that $\Line{cu} = \Line{cv}$. Let us observe
  the possible distances from $z$ to $c$. Since $c$ is a neighbour
  of $u$, $d_G(z,c)$ is in $\{k-1,k,k+1\}$. Similarly, since $c$
  is a neighbour of $v$, $d_G(z,c)$ is in $\{k,k+1,k+2\}$. So this
  distance can be either $k$ or $k+1$. Now observe that in both
  cases, this means that $z$ is in the symmetric difference of
  lines $\Line{cu}$ and $\Line{cv}$ which is a contradiction.
\end{proof}

In our definition of a good pair $uv$, we only ask for one middle
vertex $c$ to satisfy $\Line{cu} = \Line{cv}$. It turns out that
it has to be true for any common neighbour of $u$ and $v$.

\begin{lemma}
  \label{lem:eqgood}
  In any connected graph $G$, if $uv$ is a good pair then for
  every vertex $c$ between $u$ and $v$, $\Line{cu} = \Line{cv}$.
\end{lemma}

\begin{proof}
  Let $c$ be a middle vertex of $u$ and $v$, and assume that
  $\Line{cu}$ is not equal to $\Line{cv}$. Without loss of
  generality we may assume that there exists a vertex $z$ which is
  not in $\Line{cu}$ but is in $\Line{cv}$. Since $z$ is not in
  $\Line{cu}$ and $cu$ is an edge, we have $d_G(z,c) =
  d_G(z,u)$. Let us call this distance $k$.  Now since $v$ is a
  neighbour of $c$, the distance from $z$ to $v$ is in
  $\{k-1,k,k+1\}$. It cannot be equal to $k$ since $z$ is in
  $\Line{cv}$. Thus it is $k-1$ or $k+1$ and so $z$ is not in
  $\Line{uv}$. But then by Lemma~\ref{lem:distgood}, $z$ should be
  equidistant to $u$ and $v$ which is a contradiction.
\end{proof}

\subsubsection{Structural results about \{house, hole\}-free graphs}

We already mentioned that for any edge $uv$, a vertex which is
not in line $\Line{uv}$ must be equidistant to $u$ and $v$. Next
lemma gives more insight on the situation of a vertex not
belonging to a line generated by an edge.

\begin{lemma}
  \label{lem:1line}
  Let $G$ be a graph in $\HH$ and $uv$ be an edge of $G$. Then,
  for every vertex $z$ not in line $\Line{uv}$, there exists a
  common neighbour $w$ of $u$ and $v$, such that $w$ lies on a
  shortest path from $z$ to $u$ and a shortest path from $z$
  to~$v$.
\end{lemma}

\begin{proof}
  Since $z$ is not in $\Line{uv}$ and $uv$ is an edge, $z$ is
  equidistant to $u$ and $v$. Let $W$ be the set of all vertices
  that lie on both a shortest path from $z$ to $u$ and a shortest
  path from $z$ to $v$. Formally,
  \begin{equation*}
    W = I(z,u) \cap I(z,v).
  \end{equation*}
  Observe that any vertex in $W$ has its distances to $u$ and $v$
  equal. Moreover, $W$ is not empty since $z$ is in $W$. Now let
  $w$ be a vertex in $W$ with minimum distance to $u$. Let $P_u$
  be a shortest path from $u$ to $w$. There exists some integer
  $k$ such that $P_u$ is $(u_0,u_1,\ldots,u_k)$ where $u_0$ is $u$
  and $u_k$ is $w$. Similarly, let $P_v$ be a shortest path from
  $v$ to $w$. For this same integer $k$ we may describe $P_v$ as
  $(v_0,v_1,\ldots,v_k)$ such that $v_0$ is $v$ and $v_k$ is $w$.
  
  \begin{figure}[ht]
    \center
    \begin{tikzpicture}
      \node[vertex] (u) at (-.5,0) {};
      \node[left] at (u) {$u$};
      \node[vertex] (v) at (1.5,0) {};
      \node[right] at (v) {$v$};
      \draw (u)--(v);
      \node[vertex] (z) at (.5,4) {};
      \node[above] at (z) {$z$};
      \node[vertex] (w) at (.5,2) {};
      \node[right] at (w) {$w$};
      \node[vertex] (v1) at (1,1) {};
      \node[right] at (v1) {$v_{k-1}$};
      \node[vertex] (u1) at (0,1) {};
      \node[left] at (u1) {$u_{k-1}$};
      \draw[spath] (z) -- (w);
      \draw (u1) -- (w) -- (v1);
      \draw[spath] (u) -- (u1);
      \draw[spath] (v) -- (v1);
    \end{tikzpicture}
    \caption{Proof of Lemma~\ref{lem:1line}}
    \label{fig:1line}
  \end{figure}
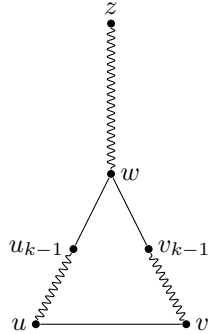

  Now since $P_u$ and $P_v$ are shortest paths, they have no
  internal chord. Moreover, by our choice of $w$, edges between
  $P_u$ and $P_v$ have to be on a same level (joining $u_i$ with
  $v_i$).

  If $k$ is 2 or more, we have either an induced hole or a house,
  which is a contradiction (see illustration on
  Figure~\ref{fig:1line}). So $k$ equals 1 and $w$ satisfies the
  conclusion of our statement.
\end{proof}

Now we state two lemmas the proof of which is neither interesting
nor enlightening. But they help for later proofs to be more
reader-friendly.

\begin{lemma}[Roof lemma]
  \label{lem:roof}
  If a graph $G$ is in $\HH$ and there is a cycle
  $C=(x_1,x_2,\ldots,x_k)$ of order at least 5 such that:
  \begin{itemize}
  \item $x_2x_k$ is an edge, and
  \item $x_1$ and $x_2$ have no other neighbours in $C$,
  \end{itemize}
  then $x_3x_k$ is an edge in $G$.
\end{lemma}
\begin{proof}
This is true when the cycle has length 5 (otherwise we have an
induced house). When the cycle is longer, assume that $x_3x_k$ is
not an edge. If $x_k$ has another chord in this cycle (to $x_i$),
we may apply the lemma on this shorter cycle
($x_1,x_2,\ldots,x_i,x_k$). Otherwise, $x_k$ has no other
neighbour. By considering a longest chord (in the sense that it
shortcuts a long part of the cycle), we either get an induced hole
or an induced house.
\end{proof}

Next lemma relies on Lemma~\ref{lem:1line}. 

\begin{lemma}
  \label{lem:C4}
 Let $G$ be a graph in $\HH$ and $C$ an induced cycle of length 4
 in $G$. If a vertex is at distance $k$ from two consecutive
 vertices of $C$, then it is at distance at most $k$ from one of
 the remaining vertices in $C$.
\end{lemma}
\begin{proof}
  Let $x_0x_1x_2x_3$ be the induced 4-cycle and assume for a
  contradiction that there is a vertex $z$ which is at distance
  $k$ from $x_0$ and $x_1$ but at distance $k+1$ from both $x_2$
  and $x_3$. By Lemma~\ref{lem:1line}, there is a vertex
  $z'$ which is in the common neighbourhood of $x_0$ and
  $x_1$. This vertex cannot be adjacent to $x_2$ or $x_3$ for
  distance reasons. So we have an induced house.
\end{proof}


\section{Structure in $\LL_1$: complete bipartite subgraphs}
\label{sec:l1}

In this section, we study the class of edges that generate the
same line. Mainly we prove that such a set of edges induce a
complete bipartite subgraph.

\begin{proposition}
  \label{prop:eqlines1}
  Let $G$ be a graph in $\HH$ and let $uv$ and $xy$ be two edges
  of $G$ such that $\Line{uv} = \Line{xy}$, this line not being
  universal. Then, either $\{u,v,x,y\}$ induces a $P_3$ (pairs
  $uv$ and $xy$ share one vertex), or $\{u,v,x,y\}$ induces a
  $C_4$ in $G$ and each edge of this $C_4$ generates the same
  line.
\end{proposition}

\begin{proof}
  Let $uv$ and $xy$ be two edges of $G$ generating the same line
  $\ell$. By Theorem~\ref{thm:eqlines} they are $\alpha$-related,
  $\beta$-related or $\gamma$-related.

  \paragraph{No possible $\gamma$-relation}

  Observe that if $uv$ and $xy$ were $\gamma$-related, then these
  would be four distinct vertices and $x$ would be in the interval
  $I(u,v)$ (see definition of $\gamma$-relation in
  Section~\ref{sub:aboulker}). This would contradict the fact
  that $uv$ is an edge of $G$.

  \paragraph{If $\alpha$-related then they induce a $P_3$}

  Assume that $uv$ and $xy$ are $\alpha$-related. If a shortest
  path in $G$ goes through all those vertices, it must visit one
  edge and then the other edge. Without loss of generality we may
  assume that there is a shortest path from $u$ to $y$ which
  starts with edge $uv$ and ends with edge $xy$. Let $k$ denotes
  the distance between $v$ and $x$. We want to prove that $k$ is
  0. As the line is not universal, and by Lemma~\ref{lem:1line},
  there is a vertex $z$ not in $\Line{uv}$ in the common
  neighbourhood of $u$ and $v$. This vertex is not in $\Line{uv}$
  so by assumption, it is not in $\Line{xy}$. Applying again
  Lemma~\ref{lem:1line}, there is a vertex $w$ in the common
  neighbourhood of $x$ and $y$ such that $w$ is in $I(z,x)$ and in
  $I(z,y)$. Let $l$ denote the distance from $z$ to $w$ (see
  Figure~\ref{fig:alpha1}). We shall prove that $k=l$. The
  distance from $u$ to $y$ is $k+2$ so the path going through $z$
  and $w$ (length $l+2$) cannot be strictly shorter. Thus, $k \leq
  l$. Similarly, since $w$ is between $z$ and $x$, the distance
  from $z$ to $x$ is $l+1$ so the path from $z$ to $x$ going
  through $v$ cannot be strictly shorter. Thus $l \leq k$. In the
  end, $k=l$.

  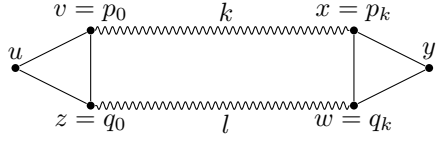
\begin{figure}[ht]
  \center
  \begin{tikzpicture}
    \node[vertex] (u) at (-3.5,0.5) {};
    \node[above] at (u) {$u$};
    \node[vertex] (v) at (-2.5,1) {};
    \node[above] at (v) {$v=p_0$};
    \draw (u)--(v);
    \node[vertex] (y) at (2,0.5) {};
    \node[above] at (y) {$y$};
    \node[vertex] (x) at (1,1) {};
    \node[above] at (x) {$x=p_k$};
    \draw (x)--(y);
    \draw[spath] (v)--(x);
    \node[vertex] (z) at (-2.5,0) {};
    \node[below] at (z) {$z=q_0$};
    \draw (u)--(z)--(v);
    \node[vertex] (w) at (1,0) {};
    \node[below] at (w) {$w=q_k$};
    \draw (x)--(w)--(y);
    \draw[spath] (w)--(z);
    \node (l) at (-0.7,0){};
    \node[below] at (l) {$l$};
    \node (k) at (-0.7,1){};
    \node[above] at (k) {$k$};
  \end{tikzpicture}
  \caption{When $uv$ and $xy$ are $\alpha$-related}\label{fig:alpha1}
\end{figure}

  Let us name vertices on a $ux$-shortest path
  $v=p_0,p_1,\ldots,p_k=x$ and similarly vertices on a
  $zw$-shortest path $z=q_0,q_1,\ldots,q_k=w$ (see
  Figure~\ref{fig:alpha1}). For distance reasons, these two paths
  must be vertex disjoint (otherwise, $z$ would be too close to
  $x$ or $v$ would be too close to $y$). Observe that vertex $y$
  has no other neighbour than $x$ and $w$ among all those
  vertices. Moreover, the only possible chords in the cycle must
  be again between vertices of the same level (otherwise $z$ would
  be in $\Line{xy}$ of $w$ would be in $\Line{uv}$). In the end,
  if $k$ is 1 or more, we have an induced house or an induced hole
  (actually an even hole). So $k = 0$ and edges $uv$ and $xy$ are
  adjacent. We have an induced $P_3$.

  \paragraph{If $\beta$-related then a $C_4$ of equivalent lines}

  Now for the last case, assume that $uv$ and $xy$ are
  $\beta$-related. Without loss of generality, we may assume that
  $(u,v,x,y)$ is a parallelogram.  Let $k$ be the distance from
  $u$ to $y$. By Property~\ref{prop:parallelogram},
  $d_G(u,y)=d_G(v,x)=k$ and $d_G(u,x) = d_G(v,y) = k+1$.
  
  We claim that $k=1$. Assume for a contradiction that $k \geq
  2$. Let $P_{uy}$ be a shortest $uy$-path We note its vertices
  $u=q_0,q_1,\ldots,q_k=y$. Similarly, let $P_{vx}$ be a shortest
  $vx$-path and let its vertices be $v=p_0,p_1,\ldots,p_k=x$. By
  definition of a parallelogram, we easily get that $P_{uy}$ and
  $P_{vx}$ must be vertex disjoint. The cycle made with $uv$,
  $P_{vx}$, $xy$ and $P_{uy}$, may have chords but only on a same
  level (of the form $p_iq_i$) for distance reasons. Actually, in
  order to prevent holes, all these chords must be present. Now,
  since $\Line{uv}$ is not universal and by Lemma~\ref{lem:1line},
  there exists a vertex $z$ not in $\Line{uv}$ such that $z$ is in
  the common neighborhood of $u$ and $v$. Now since $z$ is not in
  $\Line{xy}$, it is equidistant to $x$ and $y$. Moreover, by
  Lemma~\ref{lem:1line}, there is a vertex $w$ in the common
  neighbourhood of $x$ and $y$ which lies on a shortest $zx$-path
  and a shortest $zy$-path. All those vertices from $z$ to $w$ are
  equidistant and thus are not in $\Line{xy}$. This ensures that
  they are all new vertices. Let $l$ denote the distance from $z$
  to $y$ (and to $x$).

 To prevent a house on $\{z,u,v,p_1,q_1\}$, $z$ must have a
 neighbour in $\{p_1,q_1\}$. By symmetry, we may take $p_1$
 without loss of generality. Now this yields a path of length $k$
 from $z$ to $x$. So $l \leq k$. Moreover, through $z$ one may
 find a path from $u$ to $x$ of length $l+1$. Since $d_G(u,x) =
 k+1$ we get that $k \leq l$. Thus we have equality $k = l$.  This
 ensures that $v$ has no other neighbours in the picture (the path
 to $y$ through $z$ and $w$ is a shortest path). For similar
 reasons, $u$ also has no other neighbours among all involved
 vertices. Since $k$ is at least 2, by the roof
 lemma~\ref{lem:roof} applied on the cycle made of $uv$, $vz$,
 $P_{zw}$, $wy$ and $P_{uy}$, we must have the edge $zq_1$. But
 now, to prevent a house on $\{z,q_1,p_1,q_2,p_2\}$ we need an
 edge between $z$ and $p_2$ or $q_2$ contradicting the distance
 between $z$ and $x$ or $y$.  This proves that $k = 1$.
  
  It remains to prove that $\Line{vx} = \Line{uy} =
  \Line{uv}$. But now we have an induced $C_4$ and we may apply
  extensively Lemma~\ref{lem:C4}. If a vertex is not in
  $\Line{uv}$ then it is at some distance $k$ from $u$ and $v$ and
  at some distance $l$ from $x$ and $y$. By Lemma~\ref{lem:C4}, $k
  = l$. Therefore, this vertex is not in $\Line{vx}$ nor
  $\Line{uy}$. Reciprocally, if a vertex is not in $\Line{vx}$ it
  is at some distance $k$ from both $v$ and $x$. If it is at
  distance $k$ from $u$ or $y$, then we are in the same case as
  previously, and it is at distance $k$ from everyone, and thus is
  not in $\Line{uv}$. So we may assume it is at distance $k-1$
  from $u$. Then it is in $\Line{uv}=\Line{xy}$, so it has no
  other choice than being at distance $k-1$ from $v$ which is a
  contradiction with Lemma~\ref{lem:C4}. Hence, $\Line{vx} =
  \Line{uv}$ and similarly $\Line{uy}=\Line{uv}$.
\end{proof}

Proposition~\ref{prop:eqlines1} admits the following corollary
that we will use in the final proof of the theorem.

\begin{corollary}
  \label{coro:bipl1}
   Let $G$ be a graph in $\HH$, $\ell$ a non-universal line in
   $\LL_1$ and $F$ the set of edges generating $\ell$. Then the
   subgraph restricted to $F$ is a bipartite complete graph (and
   it is an induced subgraph).
\end{corollary}

\begin{proof}
  We prove this by induction on the size of $F$. If $F$ is a
  single edge it is trivially an induced bipartite complete
  subgraph.

  Now if the first $k$ edges of $F$ induce a bipartite complete
  subgraph (bipartition $(U,W)$). Let us pick a new edge $e=uv$ in
  $F$. By Proposition~\ref{prop:eqlines1}, it must be adjacent or
  in a $C_4$ with all previous edges. We chose preferably an edge
  incident to the current bipartition.

  Assume $e$ is incident to one of the current bipartition (let us
  say $u$ is in $U$). Then if there is another vertex $u'$ in $U$
  pick any vertex $w$ in $W$. Edges $e$ and $u'w$ must be in
  $\beta$-relation and induce a $C_4$ of $\ell$-generators. So
  that $v$ can be added to $W$ safely.

  Now, if $e$ is incident to no vertex in $(U,W)$, then it is in
  $\beta$-relation with any of the former edges. So there is a
  $C_4$ and it contains an untreated edge incident to $U$ or $W$
  which contradicts our choice of $e$.

  In the end, a class of edges generating a non-universal line
  forms an induced bipartite complete graph in $G$.
\end{proof}

\section{Structure of $\LL_2$ generators: stars}

\label{sec:l2}
In this section, we follow the same study as for previous
section. Our main result is that in the case of $C_4$-module-free
graphs (which will be our only remaining case in the end), the set
of all good pairs generating a specific line from $\LL_2$ is
arranged such that some vertex is shared by all those
pairs.

We start with the $\LL_2$-version of Lemma~\ref{lem:1line} in
order to see how vertices not in a line can be attached to a good
pair generating this line.

\begin{lemma}
  \label{lem:2line}
  Let $G$ be a \{hole\}-free connected graph and let $uv$ be a
  good pair in $G$. Then, for every vertex $z$ not in line
  $\Line{uv}$, there is a vertex $c$ between $u$ and $v$ such that
  $c$ is on both a shortest $zu$-path and a shortest $zv$-path.
\end{lemma}

\begin{proof}
  By Lemma~\ref{lem:distgood}, $z$ is equidistant from $u$ and
  $v$. Moreover this distance is at least 2 (otherwise $z$ is in
  $\Line{uv}$). Now let $W$ be the set of vertices lying on both a
  shortest $zu$-path and a shortest $zv$-path:
  \begin{equation*}
    W := I(z,u) \cap I(z,v).
  \end{equation*}
  Note that $W$ is not empty (it contains $z$) and all elements of
  $W$ are equidistant from $u$ and $v$. Now let $w$ be an element
  of $W$ with minimum distance to $u$. If this distance is 1, then
  $w$ is a middle vertex and we have our conclusion. Now assume
  that $d_G(w,u) = k$ with $k \geq 2$.  We obtain two shortest
  paths $P_{uw} = (u,u_1,\ldots,u_{k-1}, w)$ and $P_{vw} =
  (v,v_1,\ldots, v{k-1}, w)$. Observe that they have to be
  disjoint (except for $w$) otherwise it would contradict the
  minimality of $d_G(w,u)$. For the same reason the only possible
  edges between $P_{uw}$ and $P_{vw}$ are on the same level
  ($u_iv_i$ for $1 \leq i \leq k-1$). Let $c$ be a middle vertex
  of $uv$, then $wc$ is not an edge (otherwise $c$ is in $W$ and
  strictly closer to $u$ than $w$). Furthermore, $c$ has no
  neighbour in either of $P_{uw}$ and $P_{vw}$. Indeed, the only
  candidates (for distance reasons) would be $u_{k-1}$ or
  $u_{k-2}$. In the latter case, $c$ would be in $W$ and this
  would contradict the choice of $w$. In the first case, then
  $u_{k-1}$ would be at distance exactly 1 from $u$ and exactly 2
  from $v$ and thus not part of line $\Line{uv}$ but this would
  contradict Lemma~\ref{lem:distgood}.

  In the end consider the ``horizontal'' chord (of the form
  $u_iv_i$) which is closest to $u$. The cycle it forms with $ucv$
  is induced and of length at least $5$ which is a contradiction
  since $G$ is $\{hole\}$-free.
\end{proof}

Now, observe that in its very definition, the interior of a good
pair $uv$ is more than just $\{u,v\}$ (it must contain at least
one common neighbour of $u$ and $v$). Thus, such a pair cannot be
$\beta$-related (see Subsection~\ref{sub:aboulker}). So two good
pairs which generate the same line are either $\alpha$-related or
$\gamma$-related. The next proposition ensures that when we have a
$\gamma$-relation, the graph must contain a $C_4$-module. It is of
interest to us since the minimal counter-examples
for~\eqref{eq:prop} cannot contain a $C_4$-module.

\begin{proposition}\label{prop:gammal2}
  Let $G$ be a graph in $\HH$ and $uv$ and $xy$ two good pairs
  generating the same line. If $uv$ and $xy$ are $\gamma$-related,
  then $\{u,v,x,y\}$ is a $C_4$-module in $G$.
\end{proposition}

\begin{proof}
  By the definition of $\gamma$-relation, line $\Line{uv}=I(u,v)$
  and $\Line{xy} = I(x,y)$. Since $u$ and $v$ are at distance 2,
  vertices $x$ and $y$ must be common neighbours of $u$ and $v$
  (and reciprocally). Thus, those four vertices induce a
  4-cycle. Moreover, since these are both good pairs, by
  Lemma~\ref{lem:eqgood}, $\Line{ux} = \Line{vx}$ and also
  $\Line{uy} = \Line{vy}$. Similarly, $\Line{ux} = \Line{uy}$ so
  that all these edges generate the same line $\ell$ (line $\ell$
  is in $\LL_1$).

  Let us prove that those four vertices form a module. Let $z$ be
  a distinct fifth vertex.
  \begin{itemize}
  \item If $z$ is in the common neighbourhood of two consecutive
    vertices on the 4-cycle, then $z$ is not in $\ell$. But then
    it has to be adjacent to both other vertices of the 4-cycle.
  \item If $z$ is in the common neighbourhood of two opposite
    vertices of the 4-cycle (say $u$ and $v$), then it is in
    $\Line{uv}$ so it must be in $\Line{xy}$. But we know that
    $\Line{xy}=I(x,y)$, so $z$ is also adjacent to both $x$ and
    $y$.
  \item Finally, if $z$ sees exactly one vertex of the four cycle
    (say $u$). Then it must be in $\Line{uv}$ by
    Lemma~\ref{lem:distgood}, and since $\Line{uv} = I(u,v)$, it
    must be adjacent to $v$, a contradiction.
  \end{itemize}
So any such vertex $z$ is either complete or anti-complete
to $\{u,v,x,y\}$.\end{proof}

We later prove that a minimal counterexample to \eqref{eq:prop}
does not have a $C_4$-module (see Proposition~\ref{prop:red}). So
we focus on the $\alpha$-relation. We prove, as for lines in
$\LL_1$, that $\alpha$-related pairs need to intersect.

\begin{proposition}
  \label{prop:alphal2}
  Let $G$ be a graph in $\HH$ with no universal edge and two good
  pairs $uv$ and $xy$ generating the same non-universal line. If
  $uv$ and $xy$ are $\alpha$-related, then $\{u,v,x,y\}$ has
  cardinality $3$.
\end{proposition}

\begin{proof}
  By the definition of $\alpha$-relation, there is a shortest path
  going through all those vertices. Since both pairs are at
  distance 2, there are essentially two cases, either there is
  shortest path visiting the four vertices in the order $uxvy$, or
  in the order $uvxy$.

  The first case implies that $uxvy$ induces a path on 4
  vertices. Observe edge $xv$. Since it is not universal, there
  must be a vertex $z$ in the common neighbourhood of $x$ and
  $v$. But since $xy$ is a good pair, $\Line{xv} = \Line{vy}$ and
  $z$ must be a neighbour of $y$ (otherwise it is in
  $\Line{vy}$). Similarly, $z$ is a neighbour of $u$. But then
  the distance between $u$ and $y$ is 2 and the path $uxvy$ is not
  a shortest path, a contradiction.

  Now for the second case, without loss of generality, we may
  assume that there is a shortest path from $u$ to $y$ that goes
  through $u,v,x,y$ in that order. We shall prove that $v =
  x$. For this, let $k$ denote the distance between $v$ and
  $x$. Since $\Line{uv}$ is not universal, by
  Lemma~\ref{lem:2line}, there is a vertex $z$ not in
  $\Line{uv}$ and a vertex $c$ between $u$ and $v$ such that $c\in
  I(z,u)\cap I(z,v)$. We may take $z$ as a neighbour of $c$. Now
  by applying this same proposition to $z$ and $xy$, there is a
  vertex $c'$ between $x$ and $y$ such that $c'$ is in $I(z,x)$
  and in $I(z,y)$. Let $z'$ be the predecessor of $c'$ on a
  shortest $zc'$ path. We let $l$ denote the distance between $z$
  and $z'$ (see Figure \ref{fig:propalphal2}).

  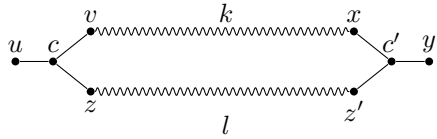
\begin{figure}[H]
    \center
    \begin{tikzpicture}
      \node[vertex] (u) at (-3.5,0.6) {};
      \node[above] at (u) {$u$};
      \node[vertex] (c) at (-3.0,0.6) {};
      \node[above] at (c) {$c$};
      
      \node[vertex] (v) at (-2.5,1) {};
      \node[above] at (v) {$v$};
      \draw (u)--(c)--(v);
      \node[vertex] (y) at (2,0.6) {};
      \node[above] at (y) {$y$};
      \node[vertex] (x) at (1,1) {};
      \node[above] at (x) {$x$};
      \node[vertex] (c1) at (1.5,0.6) {};
      \node[above] at (c1) {$c'$};
      
      \draw (x)--(c1)--(y);
      \draw[spath] (v)--(x);
      
      \node[vertex] (z) at (-2.5,0.2) {};
      \node[below] at (z) {$z$};
      \draw (z)--(c);
      
      \node[vertex] (z1) at (1,0.2) {};
      \node[below] at (z1) {$z'$};
      \draw (z1)--(c1);
      \draw[spath] (z)--(z1);
      \node (l) at (-0.7,0){};
      \node[below] at (l) {$l$};
      \node (k) at (-0.7,1){};
      \node[above] at (k) {$k$};
      
  \end{tikzpicture}
  \caption{When $u,v,x,y$ are in a shortest path in that order.
  }\label{fig:propalphal2}
\end{figure}
  Now, the shortest path from $z$ to $x$ going through $c'$ has
  length $l+2$ but there is another path from $z$ to $x$ of length $k+2$ so $l \leq   k$. 
  Then, the shortest path from $c$ to $c'$ has length $k+2$
  but there is a path through $z$ of length $l+2$ so that we may
  conclude that $k = l$.

  Observe that there is no chord from the top path to the bottom
  path. Vertical chords would create a shortcut from $z$ to
  $x$. And non-vertical chords would either put $z$ in $\Line{xy}$
  or $z'$ in $\Line{uv}$. In the end, this is a induced cycle of
  length $2k+4$. The only way to avoid a hole is to have $k = 0$.
\end{proof}

Now we can prove that in our potential minimal counter-examples, a
set of good pairs generating the same line are in special
configuration. We say that a set of good pairs is an
\emph{extended star} if all pairs share a common vertex. We call
this vertex the center.

\begin{corollary}
  \label{coro:starl2}
   If a graph $G$ in $\HH$ has no universal edge and does not
   admit a $C_4$-module, then any set $M$ of good pairs
   generating a given non-universal line is an extended star. 
\end{corollary}
 
 \begin{proof}
   Let $M = \{ (x_1,y_1), (x_2, y_2), \ldots, (x_k, y_k) \}$ such
   that for each $i$, the good pair $(x_i, y_i)$ generates the
   same line.

   Note that a single pair is an extended star by
   definition. Since the graph does not admit a $C_4$-module, any
   two good pairs are $\alpha$-related ($\beta$ and $\gamma$
   relations have been ruled out by previous
   discussion). Moreover, by Proposition~\ref{prop:alphal2}, any
   two good pairs in $\alpha$ relation share a vertex. In
   particular, we can assume $x_1 = x_2$.

   For a contradiction, assume that $M$ is not an extended star,
   then $k$ must be at least 3 and there exists some index $i$
   such that the good pair $x_iy_i$ does not involve $x_1$ (which
   is equal to $x_2$). Therefore, we may assume that $x_i = y_1$
   and $y_i = y_2$. We deduce that $d(x_1, x_i) = d(x_1, y_i) = 2$
   which implies that $x_1$ is not in the line $\Line{x_i,
     y_i}$. That is a contradiction. Hence, $M$ is an extended
   star.
 \end{proof}


\section{Sets $\LL_1$ and $\LL_2$ are disjoint}

\label{sec:l1l2}

In this section, we prove that $\LL_1$ and $\LL_2$ are disjoint
sets of lines when $G$ is a graph in $\HH$ with no universal line.

\begin{proposition}
  \label{prop:l1l2}
  Let $G$ be a graph in $\HH$ with no universal line, then the set
  of lines induced by edges and the set of lines induced by good
  pairs are disjoint sets.
\end{proposition}

\begin{proof}

  Assume there is a line $\ell$ both in $\LL_1$ and $\LL_2$, and
  let $xy$ be an edge generating $\ell$, and $uv$ a good pair
  generating $\ell$. By Theorem~\ref{thm:eqlines}, those two pairs
  must be $\alpha$-related ($\beta$-relation is excluded because
  $dist(u,v)=2$ and $\gamma$-relation because $I(x,y)$ has order 2 while
  $I(u,v)$ has order at least 3). So there is a shortest path
  containing all those vertices.

  \paragraph{Edge $xy$ is not between $u$ and $v$.}
  For a contradiction assume $y=v$ and $x$ is a middle vertex of
  $uv$. By Lemma~\ref{lem:eqgood}, $\Line{ux}$ is also equal to
  $\ell$. Now since $\ell$ is not universal, there is a vertex $z$
  in the common neighbourhood of $x$ and $y$. Thus it has to be a
  neighbour of $u$ (because it must be out of $\Line{ux}$) and for
  that reason it is in the line $\Line{uv}$ which is a
  contradiction.

  \paragraph{A shortest path from $u$ to $y$.}   
  So we may assume that there is a shortest $uy$-path going
  through $v$ before visiting $xy$. Let $P_{vx}$ be a shortest
  path between $v$ and $x$ and let $k$ be its length.  Since
  $\ell$ is not universal, there is, by Lemma~\ref{lem:2line}, a
  vertex $w$ at distance $2$ from both $u$ and $v$, having as a
  neighbour a middle vertex $c$ of $u$ and $v$.  Since $w$ is not
  in $\ell$, by Lemma~\ref{lem:1line}, there is a vertex $z$ in
  the common neighbourhood of $x$ and $y$ such that $z\in
  I(w,x)\cap I(w,y)$. Let $P_{wz}$ be a shortest path between $w$
  and $z$. Note that vertices of $P_{wz}$ are out of
  $\ell$. Hence, the graph made by $cv$, $P_{vx}$, $xy$, $yz$,
  $P_{zw}$ and $wc$ is a cycle. Moreover, $zx$ is a chord of this
  cycle and vertices $x$ and $y$ have no other neighbors in the
  cycle (otherwise that would create chords in a shortest path,
  impossible). Hence, by the Roof Lemma (Lemma~\ref{lem:roof}),
  there must be an edge from $z$ to the vertex before $x$ in
  $P_{vx}$. But then there is a shortest path from $u$ to $y$
  going through both $v$ and $z$, which put $z$ into $\Line{uv}$,
  a contradiction. 

\end{proof}


\section{Proof of main theorem}

In this last section, we provide a proof of
Theorem~\ref{thm:main}. As mentioned in the introduction, we
proceed by induction and aim to prove the stronger
statement~\eqref{eq:prop} that for every graph $G$ in $\HH$,

\begin{equation*}
   G \text{ has a universal line }\Line{uv} \text{, with
   }d_G(u,v)\leq 2\\ \text{, or at least } n \text{ lines.}
\end{equation*}

\subsection{Induction step}

\label{sec:noC4}

In this section, we prove that a minimum counter-example
to~\eqref{eq:prop} does not contain a $C_4$-module.  To ease the
presentation we call a pair of vertices $u$ and $v$ a {\em 2-pair}
if they are at distance exactly 2 (more general than good
pairs). We say that a 2-pair $uv$ is {\em universal} in $G$ if
$\Line{uv}=V$.

\begin{proposition}
  \label{prop:red}
  Let $G$ be a graph and let $x_0,x_1,x_2,x_3$ be a module of $G$
  inducing a $C_4$ (in the natural order suggested by their
  indices). If $G-x_0$ satisfies \eqref{eq:prop}, then $G$ also
  satisfies \eqref{eq:prop}.
\end{proposition}

\begin{proof}
  
  Note that the distance between a pair of vertices in $V(G)
  \setminus \{x_0\}$ is the same in $G$ and $G-\{x_0\}$. With this
  fact in mind, observe that:
  \begin{itemize}
  \item if an edge is universal in $G-x_0$, it is universal in $G$,
  \item if a 2-pair $uv$ is universal in $G-x_0$ and if $x_2$ is
    not in $\{u,v\}$, then $uv$ remains a universal 2-pair in $G$,
  \item any 2-pair involving $x_2$ in $G-x_0$ cannot be universal
    (it has to miss $x_1$ and $x_3$).
  \end{itemize}
  
  This means that if $G-x_0$ has a universal edge or a universal
  2-pair, then it is also universal in $G$. Thus, we may assume
  that $G-x_0$ has at least $|V|-1$ distinct lines. Each of these
  lines is carried over to $G$ (either by adding $x_0$ or not) so
  $G$ has at least $|V|-1$ distinct lines.

  Now, if line $\Line{x_0x_2}$ is universal, $G$
  satisfies~\eqref{eq:prop}. Otherwise, it means that there is a
  vertex not in $\Line{x_0x_2}$ and by Lemma~\ref{lem:2line},
  there must be a vertex $w$ connected to a common neighbour $c$
  of $x_0$ and $x_2$ while not being connected to $x_0$ and
  $x_2$. We claim that line $\Line{x_0w}$ is new (meaning it
  cannot be achieved in $G-x_0$). Indeed, it contains $x_0$ but it
  does not contain $x_2$. Thus it cannot be generated by two
  vertices in $V \setminus \{x_0\}$. In the end, $G$
  satisfies~\eqref{eq:prop}.
\end{proof}

\subsection{The final countdown}

\label{sec:count}

To complete the proof, we will show that all graphs in $\HH$ with
no induced $C_4$-module satisfy~\eqref{eq:prop}. Let $G$ be such a
graph on $n$ vertices. If $G$ has a universal edge or a universal
2-pair, then we are down. Actual counting needs to be performed
when there is no such universal pair.

We will only count lines in $\LL_1$ and $\LL_2$. And our proof
will rely on a discharging technique. By
Proposition~\ref{prop:l1l2} those sets $\LL_1$ and $\LL_2$ are
disjoint. Let us assign a weight of 1 to each line in $\LL_1$ and
each line in $\LL_2$. Then the total distributed weight is exactly
$|\LL_1| + |\LL_2|$.

By Corollary~\ref{coro:starl2}, for any line $\ell$ in $\LL_2$,
there is a vertex $u_{\ell}$ which is the center of the extended star formed by  all good pairs that generate $\ell$. We discharge the whole weight of such a line
$\ell$ to this vertex $u_{\ell}$. Let $C$ (for ``centers'') be the
set of vertices which received a weight in this process.

Now let us focus on any line $\ell$ in $\LL_1$. By
Corollary~\ref{coro:bipl1} the set of edges generating $\ell$
induce a complete bipartite subgraph of $G$. Let $X$ and $Y$
denote a bipartition of the involved vertices. We claim that at
most one vertex in $X$ (respectively in $Y$) is not in $C$. For a
contradiction, assume there are two vertices $u$ and $v$ in $X
\setminus C$. Clearly $u$ and $v$ form a good pair (take any
vertex of $Y$ as a middle vertex). Thus, when considering line
$\Line{uv}$ we have transfer a weight of 1 to its center which
must be either $u$ or $v$, reaching a contradiction. Now we
transfer the weight of $\ell$ in two halves, $\frac{1}{2}$ to the
vertex in $X \setminus C$ and $\frac{1}{2}$ to the vertex in $Y
\setminus C$ (if such vertices do not exist, we do not transfer
anything).

In the end, we want to prove that after this process, every vertex
of $G$ has received a weight of 1. It is clear for every vertex in
$C$ (they receive the whole weight of the corresponding line of
$\LL_2$). If a vertex is not in $C$, it receives $\frac{1}{2}$
from every line in $\LL_1$ in which it is involved. Since there is
no universal edge, every edge must support a triangle (see
Lemma~\ref{lem:1line}), and all these edges generate distinct
lines. Thus every vertex is in a triangle and is incident to at
least two edges generating distinct lines. In the end, a vertex
not in $C$ is seen at least twice when scanning lines in $\LL_1$
and thus receives at least a total weight of 1.

We have proved that after transfer, every vertex has received a
weight equal to or larger than 1. So the total initial weight is
at least $n$. This proves that $|\LL_1| + |\LL_2|$ is at least
$n$. So $G$ has at least $n$ distinct lines. This concludes the
proof of~\eqref{eq:prop} for graphs in
$\HH$. Theorem~\ref{thm:main} is a mere corollary of it.

\section*{Discussion}

In the final countdown, we proved that graphs in $\HH$ with no
induced $C_4$-module have either a universal line or sufficiently
many lines generated by edges or good pairs. One may be tempted to
remove the hassle of this $C_4$-module case and try to prove that
edges and good pairs are enough for any graph in $\HH$. It is not
the case. A counterexample is given by the graph obtained by two
disjoint 4-cycle plus a universal vertex. It has nine vertices,
$\LL_2$ is of order~2 and $\LL_1$ is of order~6.

The result presented in this paper solves Problem 3 of Chvátal’s survey~\cite{chvatal_2018}. It would interesting to extends this result to prove that Conjecture 2.3 in~\cite{aboulker_matamala_2018} holds for the class of HH-free graphs. If possible, then Theorem 2.1 in~\cite{aboulker_matamala_2018} would admit a generalization where the class of chordal graphs is replaced by the class of HH-free graphs.

\section*{Acknowledgements}

This research started from a discussion on a péniche in Lyon
during ICGT 2018. It was then continued in Santiago in May 2019
under the patronage of University of Chile and Andrés Bello
National University and completed at École Normale Supérieure in
Paris by Christmas 2019. Authors are grateful to all their hosting
institutions. 


\end{document}